\documentclass[11pt]{article}
\usepackage{epsfig,amsmath,amsthm,latexsym,graphicx,amsfonts,amssymb}
\newcommand{\beq}{\begin{equation}}
\newcommand{\eeq}{\end{equation}}

\newcommand{\R}{\mathbb R}

\newcommand{\E}{\mathbb E}

\newcommand{\htheta}{\hat{\theta}}
\newcommand{\An}{\R^n_{>}}

\newcommand{\defeq}{\stackrel{\mbox{\tiny{def}}}{=}}
\newtheorem{theorem}{Theorem}
\newtheorem{corol}[theorem]{Corollary}
\newtheorem{lemma}[theorem]{Lemma}

\newtheorem{remark}{Remark}

\begin{document}

\title{SMML estimators for 1-dimensional continuous data}
\author{James G. Dowty}
\date{\today}

\maketitle

\abstract{A method is given for calculating the strict minimum message length (SMML) estimator for 1-dimensional exponential families with continuous sufficient statistics.  A set of $n$ equations are found that the $n$ cut-points of the SMML estimator must satisfy.  These equations can be solved using Newton's method and this approach is used to produce new results and to replicate results that C. S. Wallace obtained using his boundary rules for the SMML estimator.  A rigorous proof is also given that, despite being composed of step functions, the posterior probability corresponding to the SMML estimator is a continuous function of the data.}

\section{Introduction}

The minimum message length (MML) principle \cite{wallace_boulton} is an information theoretic criterion that links data compression with statistical inference \cite{wallace}.  It has a number of useful properties and it has close connections with Kolmogorov complexity \cite{wallace_dowe}.  Using the MML principle to construct estimators is known to be NP-hard in general \cite{farr} so it is common to use approximations in practice \cite{wallace}.  The term `strict minimum message length' (SMML) is used to distinguish the exact MML criterion from these approximations.

The only known algorithm for calculating an SMML estimator is Farr's algorithm \cite{farr} which applies to data taking values in a finite set which is (in some sense) 1-dimensional.  For 1-dimensional continuous data, certain rules of thumb called {\em boundary rules} can sometimes be used for calculating the SMML estimator \cite{wallace}.  However, these rules were derived from a heuristic criterion and are not in general satisfied by the SMML estimator.  Therefore the calculation of the SMML estimator, even in the simple case of 1-dimensional continuous data, is an open problem.

This paper gives a method for calculating the SMML estimator for
a 1-dimensional exponential family of statistical models with a continuous sufficient statistic.
Section \ref{S:review} recalls the relevant definitions and fixes our notation.
Our main results appear in Section \ref{S:construction}, where we give equations that the cut-points of the SMML estimator must satisfy, show how to solve these equations with Newton's method and prove a previously unknown fact about the SMML estimator.  These results are based on certain technical lemmas whose proofs are deferred to Appendix \ref{S:proofs}.  We then apply the results of Section \ref{S:construction} to examples (in Sections \ref{S:normal} and \ref{S:exp}) before addressing some numerical issues (in Section \ref{S:numerical}).  Section \ref{S:conclusion} states our main conclusions and discusses some ideas for further research.

\section{The SMML estimator}
\label{S:review}

In order to define our notation, this section briefly recalls the definition of the SMML estimator for a 1-dimensional exponential family of statistical models with a continuous sufficient statistic.

Let the exponential family have support $\mathcal{X}$ and natural parameter space $\Theta$ and assume that both are open, connected subsets of $\R$.  For each $\theta \in \Theta$, let $f(x | \theta)$ be the probability density function (PDF) on $\mathcal{X}$ given by
\begin{equation}\label{E:exp_family}
f(x | \theta) \defeq \exp(x \theta - \psi(\theta)) h(x)
\end{equation}
for any $x \in \mathcal{X}$, where $\psi$ and $h$ are given functions with $h$ strictly positive everywhere on $\mathcal{X}$.  If $\pi$ is a Bayesian prior on $\Theta$ then we define the {\em marginal PDF} $r$ to be given by
$$ r(x) \defeq \int_\Theta \pi(\theta) f(x | \theta) d\theta $$
for any $x \in \mathcal{X}$, and $r(x)= 0$ elsewhere.  We make the technical assumption that the first moment of $r$ exists.

For the 1-dimensional case considered above, the SMML estimator with $n$ cut-points is defined as follows \cite{wallace}.  Suppose we are given an integer $n \ge 1$ and real numbers $a_1 < \ldots < a_n$ in $\mathcal{X}$ (the cut-points) as well as $\htheta_0, \htheta_1, \ldots, \htheta_n \in \Theta$ (the assertions) and $q_0, q_1, \ldots q_n \in \R$ (the coding probabilities for the assertions) so that $1 = q_0+ q_1+ \ldots +q_n$ and each $q_i > 0$.  Then for each $i = 0, 1, \ldots, n$, define $U_i$ to be the interval $U_i = [a_i,a_{i+1})$ where $a_0$ and $a_{n+1}$ are the boundaries of $\mathcal{X}$, e.g. if $\mathcal{X} = \R$ then $a_0 = -\infty$ and $a_{n+1} = \infty$.  Let $\htheta$ and $q$ be the step functions given by $\htheta(x) \defeq \htheta_i$ and $q(x) \defeq q_i$ where $i$ is the unique integer for which $x \in U_i$.  If we discretize the data space $\mathcal{X}$ to a lattice then there is a $2$-part coding of the data which has expected length
\begin{equation}\label{E:I1}
I_1 = - \E[ \log(q(X) f(X | \htheta(X) ))]
\end{equation}
plus a constant which only depends on the width of the lattice, where $X$ is a random variable with PDF $r$, written $X \sim r$.  Then an SMML estimator with $n$ cut-points is a function $\htheta(x)$ which minimizes $I_1$ out of all estimators of this form.

This minimality condition can be used to solve for the assertions and the coding probabilities in terms of the cut-points.  Let $\mu: \Theta \to \R$ be the function $\mu(\theta) \defeq \int_\R x f(x | \theta) dx$ which relates the natural parametrization of the exponential family to the expectation parametrization.  By a standard result for exponential families (e.g. see Theorem 2.2.1 of \cite{kass}), $\psi$ is infinitely differentiable, $\mu = \psi^\prime$ and $\mu$ has an infinitely differentiable inverse.  Then it is not too hard to show (see R2 and R3 on pages 155-156 and 168-169 of \cite{wallace}), for each $i = 0, 1, \ldots, n$, that
\begin{equation}\label{E:q}
q_i = \int_{U_i} r(x) dx
\end{equation}
and
\begin{equation}\label{E:htheta}
\htheta_i = \mu^{-1}\left( \frac{1}{q_i} \int_{U_i} x r(x) dx \right).
\end{equation}
So (\ref{E:q}) says that $q_i$ is the mass of $U_i$ and (\ref{E:htheta}) says that the centre of mass of $U_i$ is the expectation parameter corresponding to $\htheta_i$.

Note that an SMML estimator with $n$ cut-points might not exist or might not be unique in general.  However, we will often refer to `the' SMML estimator when discussing this estimator informally.

\section{Constructing the SMML estimator}
\label{S:construction}

This section describes our construction of the SMML estimator.  This construction is given in terms of the natural parametrization of the exponential family but this determines the SMML estimator in general since this estimator transforms simply under reparametrization.

Using (\ref{E:q}), (\ref{E:htheta}) and the fact that $U_i = [a_i,a_{i+1})$, we can consider $q_i$ and $\htheta_i$ to be functions of the cut-points $a \defeq (a_1, \ldots, a_n) \in \R^n$.  Then $I_1$ becomes a function solely of $a$ and each SMML estimator with $n$ cut-points corresponds to a value of $a \in \R^n$ which minimizes this function $I_1(a)$.  But $r$ is continuous so $q_i(a)$ and $\htheta_i(a)$ are continuously differentiable functions of $a$, hence so is $I_1(a)$ by (\ref{E:I1_finitesum}) below.  Then since $I_1(a)$ is defined on the open subset
$$\An \defeq \{ a \in \mathcal{X}^n \mid a_1 < \ldots < a_n \}$$
of $\R^n$, its gradient vanishes at its minimum (if a minimum exists, i.e. if an SMML estimator with $n$ cut-points exists).  For each $j = 1, \ldots, n$ we therefore have an equation
\begin{equation}\label{E:deriv_simple}
    \frac{\partial I_1}{\partial a_j} = 0
\end{equation}
which is satisfied at any $a\in \An$ corresponding to an SMML estimator.  These $n$ equations can then be used to solve for the $n$ unknowns $a_1, \ldots, a_n$, giving the corresponding SMML estimator by (\ref{E:q}) and (\ref{E:htheta}).

The next lemma therefore calculates the partial derivatives which appear in (\ref{E:deriv_simple}).

\begin{lemma}
\label{L:grad}
Let $q_i$ and $\htheta_i$ be the functions of $a\in \An$ given by (\ref{E:q}) and (\ref{E:htheta}) and let $C$ be the constant $- \int_\mathcal{X} r(x) \log h(x) dx$.  Then
\begin{equation}\label{E:I1_finitesum}
I_1(a) = C - \sum_{i=0}^n  q_i \left( \log q_i + \htheta_i \mu(\htheta_i) - \psi(\htheta_i) \right)
\end{equation}
and, for $j = 1, \ldots, n$,
\begin{equation}\label{E:deriv}
\frac{\partial I_1}{\partial a_j} =
r(a_j) \log \left(  \frac{q_j f(a_j | \htheta_j)}{q_{j-1} f(a_j | \htheta_{j-1})} \right).
\end{equation}
\end{lemma}

\begin{proof}
See Appendix \ref{S:proofs}.
\end{proof}

Note that the numerator and denominator in the logarithm of (\ref{E:deriv}) are, respectively, the limits of $q(x) f(x | \htheta(x))$ as $x$ approaches $a_j$ from above and below.  Therefore (\ref{E:deriv_simple}) is exactly the condition which ensures that $q(x) f(x | \htheta(x))$ is a continuous function of $x$ at $x=a_j$.  So even though $q(x)$ and $\htheta(x)$ are step functions, we have proved the following.

\begin{corol} \label{C:cont}
For the SMML estimator, $q(x) f(x | \htheta(x))$ is a continuous function of $x$.
\end{corol}

Now, let $G:\An \to \R^n$ be the function whose $j^{th}$ co-ordinate is given by
\begin{equation}\label{E:G}
    G_j(a) = \log \left(  \frac{q_j f(a_j | \htheta_j)}{q_{j-1} f(a_j | \htheta_{j-1})} \right)
\end{equation}
for any $a \in \An$.  By Lemma \ref{L:grad},
$$ \frac{\partial I_1}{\partial a_j} =  r(a_j) G_j(a),$$
so since $r(a_j)$ is never zero, solving the system of equations (\ref{E:deriv_simple}) is equivalent to the
simpler and numerically better-behaved problem of finding the zeroes of the function $G:\An \to \R^n$.
We will use Newton's method to find the zeroes of $G$ so the next lemma calculates the Jacobian matrix of $G$ and shows that it is sparse.

\begin{lemma}
\label{L:jacobian}
For $j,k = 1, \ldots, n$,
\begin{eqnarray*}
\frac{\partial G_j}{\partial a_k} &=& 0 \mbox{ if $|j-k| > 1$.} \\
\frac{\partial G_j}{\partial a_{j-1}} &=&
  \frac{r(a_{j-1})}{q_{j-1}} \left( 1 + \frac{(a_{j-1} - \mu(\htheta_{j-1}))(a_j - \mu(\htheta_{j-1}))} {\mu^\prime(\htheta_{j-1})} \right)
  \mbox{ if $j \not= 1$} \\
\frac{\partial G_j}{\partial a_j} &=& \htheta_j - \htheta_{j-1} -
  \frac{r(a_j)}{q_j} \left( 1 + \frac{(a_j - \mu(\htheta_j))^2} {\mu^\prime(\htheta_j)}\right) -
  \frac{r(a_j)}{q_{j-1}} \left( 1 + \frac{(a_j - \mu(\htheta_{j-1}))^2} {\mu^\prime(\htheta_{j-1})}\right) \\
\frac{\partial G_j}{\partial a_{j+1}} &=&
   \frac{r(a_{j+1})}{q_j} \left( 1 + \frac{(a_j - \mu(\htheta_j))(a_{j+1} - \mu(\htheta_j))} {\mu^\prime(\htheta_j)} \right) \mbox{ if $j \not= n$}.
\end{eqnarray*}
\end{lemma}

\begin{proof}
See Appendix \ref{S:proofs}.
\end{proof}

Note that $\mu^\prime(\theta)$ is the variance of the distribution (\ref{E:exp_family}) for any $\theta \in \Theta$ (e.g. see Theorem 2.2.1 of \cite{kass}).

\begin{remark}
\label{R:global_min}
A global minimum of $I_1:\An \to \R$ is (the set of cut-points of) an SMML estimator, but solutions to the system of equations (\ref{E:deriv_simple}) are only critical points of $I_1$.  We can use Lemma \ref{L:jacobian} to check if a solution to (\ref{E:deriv_simple}) is a local minimum, but these might not be global minima.
\end{remark}

\begin{remark}
If an SMML estimator with $n$ cut-points $a^{(n)} \in \An$ exists for each $n$ then $I_1(a^{(n)})$ is a non-increasing function of $n$.  To see this, note that $I_1(a^{(n)}) \le I_1(a)$ for every $a \in \An$ since $a^{(n)}$ is a global minimum of $I_1:\An \to \R$.  Also, there exist $a \in \An$ with $I_1(a)$ arbitrarily close to $I_1(a^{(n-1)})$, e.g. take $a$ to be $a^{(n-1)}$ but with an extra cut-point close to one of the cut-points of $a^{(n-1)}$ and use (\ref{E:I1_finitesum}).  Therefore $I_1(a^{(n)}) \le I_1(a) = I_1(a^{(n-1)}) + \epsilon$ for every $\epsilon > 0$, so $I_1(a^{(n)}) \le I_1(a^{(n-1)})$.
\end{remark}

\begin{remark}
There is a one-to-one map between the set of possible cut-points $\An$ and the set of all $p \in \R^n$ with $0 < p_1 < \ldots < p_n < 1$, given by $p_i = R(a_i)$ where $R(x) = \int_\infty^x r(\xi) d\xi$ is the marginal cumulative distribution function.  So we can consider $a$ and hence $I_1$ to be a function of $p$ alone, in which case $\frac{\partial I_1}{\partial p_j} =  G_j$ by (\ref{E:deriv}) and the chain rule, since the Jacobian of the transformation $a \mapsto p$ is the diagonal matrix with entries $r(a_1), \ldots, r(a_n)$.  Parameterizing the cut-points in terms of $p$ has several advantages (e.g. $q_i$ is given by the simple formula $q_i = p_{i+1} - p_i$), but we will not pursue this parametrization here.
\end{remark}

\section{Normal data with known variance and a normal prior}
\label{S:normal}

We now apply the work of the previous section to a simple case.  Each set of cut-points $a$ in this section gives a local minimum of $I_1(a)$ but not necessarily a global minimum, so we will refer to these as `likely SMML estimators' to indicate that they are likely but not guaranteed to be SMML estimators (see Remark \ref{R:global_min}).

Let $\mathcal{X} = \Theta = \R$ and choose a normal prior on $\Theta$ with variance $\alpha^2$, i.e. $\theta \sim N(0,\alpha^2)$.  Let the data $X$ given $\theta$ be normally distributed with mean $\theta$ and variance $1$, i.e. $(X|\theta) \sim N(\theta,1)$.  For example, if $Y_1, \ldots, Y_m$ are independent and all distributed according to $N(\theta \frac{\sqrt{m}}{\sigma},\sigma^2)$, where $\sigma$ is known, then $(X|\theta) \defeq \frac{\sigma}{\sqrt{m}} \overline{Y}$ is a minimal sufficient statistic for $Y_1, \ldots, Y_m$ and $(X|\theta) \sim N(\theta,1)$.

The PDF of $X$ given $\theta$ is of the form (\ref{E:exp_family}) with $\psi(\theta) = \frac{1}{2} \theta^2$ and $h(x) = \frac{1}{\sqrt{2 \pi}} e^{-x^2 / 2}$.  As noted earlier, $\mu(\theta) = \psi^\prime(\theta)$, so $\mu$ is the identity map.  Also, it is not hard to show that the data $X$ (not conditioned on $\theta$) is distributed as $X \sim N(0, \beta^2)$ where $\beta = \sqrt{1+\alpha^2}$, so
$$ r(x) = \frac{1}{\beta \sqrt{2 \pi}} \exp\left(-\frac{x^2}{2 \beta^2}\right).$$

Table \ref{T:normalnormal} gives, for various numbers $n$ of cut-points, the non-negative cut-points $b_1, \ldots, b_k$ of the likely SMML estimator when $\alpha = 2$.  The bottom line of Table \ref{T:normalnormal} ($n=16$) corresponds to the `exact SMML' column in Table 3.2 on page 176 of Wallace \cite{wallace} and it agrees with this column except for Wallace's last entry, which he says is `not correct'.

Wallace generated his results using an unspecified iterative procedure which combined his boundary rules and (\ref{E:htheta}), even though he says these are `incompatible'.  Due to this incompatibility, it is maybe not surprising that the boundary rules are not satisfied for the likely SMML estimators given in Table \ref{T:normalnormal}, though this makes the close agreement between his results and ours even more surprising.  It is not clear what connections exist between the system of equations (\ref{E:deriv_simple}) and Wallace's boundary rules, but there does not seem to be a simple connection.

\begin{table}
\centering
\begin{tabular}{|l|l|llllllll|}
\hline
$n$ & $I_1-I_0$ & $b_1$ & $b_2$ & $b_3$ & $b_4$ & $b_5$ & $b_6$ & $b_7$ & $b_8$ \\
\hline
$1$ & $0.2968787967$ & $0.0000$ & $-$ & $-$ & $-$ & $-$ & $-$ & $-$ & $-$ \\
$2$ & $0.1848522963$ & $1.9740$ & $-$ & $-$ & $-$ & $-$ & $-$ & $-$ & $-$ \\
$3$ & $0.1756409558$ & $0.0000$ & $3.8977$ & $-$ & $-$ & $-$ & $-$ & $-$ & $-$ \\
$4$ & $0.1753131831$ & $1.9203$ & $5.9799$ & $-$ & $-$ & $-$ & $-$ & $-$ & $-$ \\
$5$ & $0.1753126143$ & $1.9044$ & $5.9619$ & $10.8610$ & $-$ & $-$ & $-$ & $-$ & $-$ \\
$6$ & $0.1753120750$ & $1.9203$ & $5.9797$ & $10.8840$ & $-$ & $-$ & $-$ & $-$ & $-$ \\
$7$ & $0.1753120750$ & $1.9203$ & $5.9797$ & $10.8840$ & $17.5442$ & $-$ & $-$ & $-$ & $-$ \\
$8$ & $0.1753120750$ & $1.9203$ & $5.9797$ & $10.8840$ & $17.5442$ & $-$ & $-$ & $-$ & $-$ \\
$9$ & $0.1753120750$ & $1.9203$ & $5.9797$ & $10.8840$ & $17.5442$ & $27.1130$ & $-$ & $-$ & $-$ \\
$10$ & $0.1753120750$ & $1.9203$ & $5.9797$ & $10.8840$ & $17.5442$ & $27.1130$ & $-$ & $-$ & $-$ \\
$11$ & $0.1753120750$ & $1.9203$ & $5.9797$ & $10.8840$ & $17.5442$ & $27.1130$ & $41.1964$ & $-$ & $-$ \\
$12$ & $0.1753120750$ & $1.9203$ & $5.9797$ & $10.8840$ & $17.5442$ & $27.1130$ & $41.1964$ & $-$ & $-$ \\
$13$ & $0.1753120750$ & $1.9203$ & $5.9797$ & $10.8840$ & $17.5442$ & $27.1130$ & $41.1964$ & $62.1447$ & $-$ \\
$14$ & $0.1753120750$ & $1.9203$ & $5.9797$ & $10.8840$ & $17.5442$ & $27.1130$ & $41.1964$ & $62.1447$ & $-$ \\
$15$ & $0.1753120750$ & $1.9203$ & $5.9797$ & $10.8840$ & $17.5442$ & $27.1130$ & $41.1964$ & $62.1447$ & $93.4500$ \\
$16$ & $0.1753120750$ & $1.9203$ & $5.9797$ & $10.8840$ & $17.5442$ & $27.1130$ & $41.1964$ & $62.1447$ & $93.4500$ \\
\hline
\end{tabular}
\caption{For various numbers $n$ of cut-points, the difference $I_1-I_0$ in expected code-lengths of the one- and two-part codes as well as the non-negative cut-points $b_1, \ldots, b_k$ of the likely SMML estimator. }
\label{T:normalnormal}
\end{table}

The SMML estimator seems to be unique and symmetric about $0$ when $n$ is $1$ or $3$ or $n$ is even, so Table \ref{T:normalnormal} determines the likely SMML estimator in these cases, e.g. $a = (-b_k, \ldots, -b_1, b_1, \ldots, b_k)$ if $n = 2k$.  For odd $n \ge 5$ there are two likely SMML estimators, e.g. if $n=5$ the two estimators have cut-points
$$a = (-5.9978, -1.9362, 1.9044, 5.9619, 10.8610)$$
and
$$a = (-10.8610, -5.9619, -1.9044, 1.9362, 5.9978).$$
For odd $n \ge 7$, each negative cut-point is minus one of the positive cut-points (to four decimal places), e.g. when $n=7$ the cut-points are
$$a = (-10.8840,-5.9797,-1.9203,1.9203,5.9797,10.8840,17.5442),$$
or the negative of this in the reverse order.

Table \ref{T:normalnormal} also gives the difference $I_1-I_0$ in expected code-lengths of the one- and two-part codes, where $I_0 = -\int_\mathcal{X} r(x) \log r(x) dx$.  Note that increasing the number of cut-points beyond $n=6$ improves the expected code-length by less than $10^{-10}$, so $n=6$ cut-points are probably sufficient for most practical applications.  Also note that, to four decimal places, the set of cut-points of each likely SMML estimator with $6 \le n \le 16$ is just a subset of the cut-points for the likely SMML estimator when $n=16$.  This is probably due to the fact that more than $n=6$ cut-points makes very little difference to $I_1$ and hence has little impact on the placement of the existing cut-points.

Of theoretical interest, there is a local minimum of $I_1$ at $a = (-5.9978, -1.9362,  1.9044,  5.9619, 10.8610, 17.5118)$ which is not a global minimum, since $I_1 - I_0 = 0.1753126143$ for this set of cut-points and this is is larger than $I_1- I_0$ for the cut-points given in Table \ref{T:normalnormal} for $n=6$.  So this is a counter-example to idea that all local minima of $I_1$ correspond to SMML estimators.

Figure \ref{F:even} shows the cut-points and the graphs of $D_0(x)$, $D_1(x)$ and $r(x)$ corresponding to the likely SMML estimator when $n=6$, where $D_0(x) = -r(x) \log r(x)$ and $D_1(x) = -r(x) \log(q(x) f(x| \htheta(x) ))$ so that $I_0 = \int_\mathcal{X} D_0(x) dx$ and $I_1 = \int_\mathcal{X} D_1(x) dx$.  Note that the continuity of $D_1(x)$, which is guaranteed by Corollary \ref{C:cont}, is consistent with this figure.

\begin{figure}
  \centering
  \includegraphics[width=10cm]{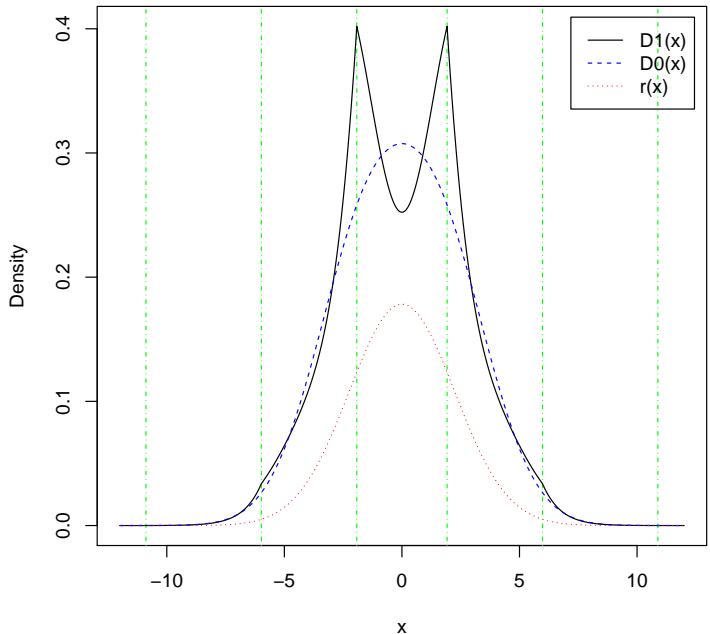} \\
  \caption{For the likely SMML estimator with $6$ cut-points, the graphs of $D_0(x)$ (dashed line), $D_1(x)$ (solid line) and $r(x)$ (dotted line) and the corresponding cut-points (vertical lines).}
  \label{F:even}
\end{figure}

\section{Exponential data with a gamma prior}
\label{S:exp}

In this section, we apply the results of Section \ref{S:construction} to exponential data with a gamma prior.

For exponential data, $f(x \mid \theta)$ is of the form (\ref{E:exp_family}) if we choose $\mathcal{X} = [0,\infty)$,  $\Theta = (-\infty, 0)$, $\psi(\theta) = -\log(-\theta)$ and $h(x) = 1$ (though note that exponential distributions are usually parameterized in terms of the rate $-\theta$).  Hence the corresponding mean is $\mu(\theta) = \psi^\prime(\theta) = -1/\theta$.  Choose a gamma prior for $-\theta$ with shape and rate parameters $\alpha>0$ and $\beta > 0$ (respectively) so that
$$ \pi(\theta) = \frac{\beta^\alpha}{\Gamma(\alpha)} |\theta|^{\alpha -1} e^{\beta |\theta|}. $$
Then the data $X$ has marginal PDF
$$ r(x) = \frac{\alpha}{\beta} \left( 1+ \frac{x}{\beta} \right)^{-\alpha - 1} ,$$
i.e. $X$ has a Lomax distribution (equivalently, $1+X/\beta$ has a Pareto distribution).  In order to satisfy our technical condition that the first moment of $r$ should exist we need to additionally assume that $\alpha > 1$.  Note that for exponential data with an exponential prior ($\alpha=1$), the expectation defining $I_1$ in (\ref{E:I1}) does not in general exist (see (\ref{E:I1_calc})), so the SMML estimator is not defined in this case.

Table \ref{T:expgamma} gives the cut-points for the likely SMML estimator when $\alpha = 2$ and $\beta = 1$.  In contrast to the normal-normal case of Section \ref{S:normal}, for exponential data and a gamma prior it seems that the SMML estimator is unique and that all local minima of $I_1(a)$ are global minima.

\begin{table}
\centering
\begin{tabular}{|l|l|lllll|}
\hline
$n$ & $I_1-I_0$ & $a_1$ & $a_2$ & $a_3$ & $a_4$ & $a_5$ \\
\hline
$1$ & $0.0589128612$ & $4.49$ & $-$ & $-$ & $-$ & $-$ \\
$2$ & $0.0579045079$ & $4.42$ & $80.55$ & $-$ & $-$ & $-$ \\
$3$ & $0.0579008163$ & $4.42$ & $80.17$ & $1380.63$ & $-$ & $-$ \\
$4$ & $0.0579008036$ & $4.42$ & $80.17$ & $1374.66$ & $23597.96$ & $-$ \\
$5$ & $0.0579008036$ & $4.42$ & $80.17$ & $1374.64$ & $23496.46$ & $403274.23$ \\
\hline
\end{tabular}
\caption{For various numbers $n$ of cut-points, the difference $I_1-I_0$ in expected code-lengths of the one- and two-part codes as well as the cut-points $a_1, \ldots, a_n$ of the likely SMML estimator. }
\label{T:expgamma}
\end{table}

\section{Numerical considerations}
\label{S:numerical}

To construct the SMML estimator we might have to consider cut-points $a_1, \ldots, a_n$ which are far outside the likely range of the data, so some of the corresponding values of $r(a_j)$ and $q_j$ might be extremely small, smaller even than machine precision.  This section briefly discusses some simple and effective solutions to the numerical problems that this causes.

By (\ref{E:exp_family}), the $j^{th}$ co-ordinate of $G:\An \to \R^n$ is given by
\begin{equation}\label{E:G_simpler}
G_j(a) = \left(  \log q_j  + a_j \htheta_j - \psi(\htheta_j) \right) -
\left(  \log q_{j-1} + a_j \htheta_{j-1} - \psi(\htheta_{j-1}) \right)
\end{equation}
for any $a \in \An$.  For any $c \in \mathcal{X}$, let $\tilde{r}_c(x) \defeq r(x)/r(c)$, so by (\ref{E:q}) we have
$$ \log q_j = \log\left( \int_{U_j} r(x)dx \right) = \log\left( r(c) \int_{U_j} \tilde{r}_c(x)dx \right)
= \log r(c) + \log\left(\int_{U_j} \tilde{r}_c(x)dx \right).$$
By choosing $c$ appropriately, all terms in the right hand side of this expression can be calculated numerically to a high degree of precision (for many functions $r(x)$).  For example, with $r(x)$ as in Section \ref{S:normal}, $\tilde{r}_c(x) = \exp(-(x^2-c^2)/2\beta^2)$, so taking $c=a_j$ we have
$$ \log q_j = -\frac{a_j^2}{2 \beta^2} - \frac{1}{2}\log(2 \pi \beta^2) + \log\left( \int_{a_j}^{a_{j+1}} \exp\left[{\frac{-(x^2-a_j^2)}{2\beta^2}}\right]dx \right)$$
which is numerically well-behaved even for large $a_j$.  Also, by (\ref{E:htheta}) we have
$$ \htheta_j = \mu^{-1}\left( \frac{1}{q_j} \int_{U_j} x r(x) dx \right) =  \mu^{-1}\left( \frac{\int_{U_j} x r(x) dx}{\int_{U_j} r(x) dx} \right) =  \mu^{-1}\left( \frac{\int_{U_j} x \tilde{r}_c(x) dx}{\int_{U_j} \tilde{r}_c(x) dx} \right)$$
and the right-hand side is again numerically well-behaved for some choice of $c$, e.g. with $r(x)$ as in Section \ref{S:normal} we could take $c=a_j$ if $a_j>0$ and $c=a_{j+1}$ if $a_j<0$.  This shows that high-precision numerical calculations of $I_1$ and $G$ are possible, even when $q_j$ is smaller than machine precision.

We also note that $r(a_i)$ and $q_j$ only appear in Lemma \ref{L:jacobian} as ratios of each other.  We can calculate $r(a_{j+1})/q_j$ by evaluating the right-hand side of
$$ \frac{r(a_{j+1})}{q_j} = \frac{r(a_{j+1})}{\int_{U_j} r(x) dx} =
\frac{\tilde{r}_c(a_{j+1})}{\int_{U_j} \tilde{r}_c(x) dx}$$
and this is numerically well-behaved for appropriate $c$.  Other ratios $r(a_i)/q_j$ can be calculated similarly so  the Jacobian matrix of $G$ can also be calculated numerically.

\section{Conclusions and extensions}
\label{S:conclusion}

In the context of 1-dimensional exponential families with continuous sufficient statistics, we have found equations that the cut-points of the SMML estimator must satisfy.  As a corollary, we proved that the posterior probability $q(x) f(x | \htheta(x))$ corresponding to the SMML estimator is a continuous function of $x$, despite being composed of step functions.  We also solved these equations for a particular example using Newton's method.  Our approach is very simple but it solves an outstanding problem in information theory which previously could only be attempted with rules of thumb like Wallace's boundary rules.

Focussing on the case of continuous data allowed us to use calculus to solve the optimization problem defining the SMML estimator.  Restricting to $1$-dimensional data allowed us to assume a particular form (intervals) for the shape of the regions defining the SMML estimator.  Therefore our results probably generalize fairly easily to non-exponential families with $1$-dimensional sufficient statistics.  It is also possible that they will generalize to higher-dimensional continuous data, if the regions which define the SMML estimator are assumed to be convex polygons (or any other shapes whose configuration space is a manifold).

Many questions about SMML estimators for continuous data remain unanswered, even in the simple, $1$-dimensional case considered here.  Does an SMML estimator with a given number of cut-points always exist?  Is the SMML estimator with a given number of cut-points unique for positive data?  Does the system of equations (\ref{E:deriv_simple}) have a finite number of solutions?  If the data is restricted to a compact (i.e. finite and closed) interval then is there an upper bound to the number of cut-points that an SMML estimator can have?

An affirmative answer to the last two questions would open the possibility of developing a rigorous algorithm to find all SMML estimators with a given number of cut-points (by finding all solutions to (\ref{E:deriv_simple}) and outputting those with the lowest $I_1$) and a continuous analogue of Farr's algorithm \cite{farr} for positive data.

\appendix
\section{Proofs of technical lemmas}
\label{S:proofs}

This appendix contains the proofs of our main technical lemmas.  We begin with a calculation which will be used in both proofs.

\begin{lemma} \label{L:some_derivs}
Let $q_i$ and $\htheta_i$ be the functions of the cut-points $a$ given by (\ref{E:q}) and (\ref{E:htheta}). Then for $i = 0,1, \ldots, n$ and $k = 1, \ldots,n$,
$$\frac{\partial q_i}{\partial a_k} = \epsilon \, r(a_k) $$
and
$$\frac{\partial \htheta_i}{\partial a_k} =
\epsilon \frac{r(a_k)}{q_i \mu^\prime(\htheta_i)} \left(a_k - \mu(\htheta_i)\right)
$$
where
$$ \epsilon =
\left\{
  \begin{array}{ll}
    -1 & \mbox{if $k = i$;} \\
    1 & \mbox{if $k = i+1$;} \\
    0 & \mbox{otherwise.}
  \end{array}
\right.
$$
\end{lemma}

\begin{proof}
Let $R$ be the marginal cumulative distribution function of the data given by $R(x) = \int_{-\infty}^x r(\xi) d\xi$ for any $x \in \R$.  Then by (\ref{E:q}), $q_i = R(a_{i+1}) - R(a_i)$ so $\frac{\partial q_i}{\partial a_i} = -r(a_i)$, $\frac{\partial q_i}{\partial a_{i+1}} = r(a_{i+1})$ and $\frac{\partial q_i}{\partial a_k} = 0$ unless $k = i, i+1$.

Now, let $M(x) = \int_{-\infty}^x \xi r(\xi) d\xi$ for any $x \in \R$ so that $q_i \mu(\htheta_i) = M(a_{i+1}) - M(a_i)$ by (\ref{E:htheta}).  Differentiating this equation with respect to $a_i$ gives
$$ -r(a_i) \mu(\htheta_i) + q_i \mu^\prime(\htheta_i) \frac{\partial \htheta_i}{\partial a_i}
=  -a_i r(a_i) $$
so by rearranging we have
$$\frac{\partial \htheta_i}{\partial a_i} = \frac{r(a_i)}{q_i \mu^\prime(\htheta_i)} \left(\mu(\htheta_i) - a_i\right).$$
The cases $k=i+1$ and $k \not=i,i+1$ can be handled similarly.
\end{proof}

\begin{proof}[Proof of Lemma \ref{L:grad}]
As in the statement, let $C$ be the constant $- \int_\mathcal{X} r(x) \log h(x) dx$.  From (\ref{E:I1}) we have
\begin{eqnarray}
I_1 &=&  - \int_\mathcal{X} r(x) \log(q(x) f(x| \htheta(x) )) dx \nonumber \\
&=& - \sum_{i=0}^n  \int_{U_i} r(x) \log(q(x) f(x| \htheta(x) )) dx  \nonumber \\
&=& - \sum_{i=0}^n  \int_{U_i} r(x) \log(q_i f(x| \htheta_i )) dx  \nonumber \\
&=& C - \sum_{i=0}^n  \int_{U_i} r(x) [ \log q_i +  x \htheta_i - \psi(\htheta_i)] dx \mbox{ by (\ref{E:exp_family})} \nonumber \\
&=& C + \sum_{i=0}^n  \left( - q_i \log q_i  + q_i \psi(\htheta_i) - \htheta_i \int_{U_i} x r(x) dx \right) \label{E:I1_calc}
\end{eqnarray}
and (\ref{E:I1_finitesum}) follows by (\ref{E:htheta}).  Then by (\ref{E:I1_calc}) and Lemma \ref{L:some_derivs},
\begin{eqnarray*}
\frac{\partial I_1}{\partial a_j} &=&  \sum_{i=j-1}^j \frac{\partial ~}{\partial a_j} \left( - q_i \log q_i  + q_i \psi(\htheta_i) - \htheta_i \int_{U_i} x r(x) dx \right) \\
&=&  r(a_j) \left(  - \log q_{j-1} - 1 + \psi(\htheta_{j-1}) - a_j \htheta_{j-1} \right) +
     \frac{\partial \htheta_{j-1}}{\partial a_j}
     \left(  q_{j-1} \psi^\prime(\htheta_{j-1}) - \int_{U_{j-1}} x r(x) dx \right)  \\
&&  + r(a_j) \left(  \log q_j + 1 - \psi(\htheta_j) + a_j \htheta_j \right) +
     \frac{\partial \htheta_j}{\partial a_j}
     \left(  q_j \psi^\prime(\htheta_j) - \int_{U_j} x r(x) dx \right).
\end{eqnarray*}
Now, $\int_{U_j} x r(x) dx = q_j \mu(\htheta_j)$ by (\ref{E:htheta}) and $\psi^\prime(\theta) = \mu(\theta)$ by Theorem 2.2.1 of \cite{kass}.  Therefore the bracketed expression multiplying $\frac{\partial \htheta_j}{\partial a_j}$ vanishes, as does the expression multiplying $\frac{\partial \htheta_{j-1}}{\partial a_j}$, so we have
\begin{equation}\label{E:deriv2}
\frac{\partial I_1}{\partial a_j} =
r(a_j) \left(  \log q_j  + a_j \htheta_j - \psi(\htheta_j) \right) -
r(a_j) \left(  \log q_{j-1} + a_j \htheta_{j-1} - \psi(\htheta_{j-1}) \right)
\end{equation}
and the lemma follows from (\ref{E:exp_family}).
\end{proof}

\begin{proof}[Proof of Lemma \ref{L:jacobian}]
By (\ref{E:exp_family}) and (\ref{E:G}),
\begin{equation}\label{E:G_simpler_again}
G_j(a) = \left(  \log q_j  + a_j \htheta_j - \psi(\htheta_j) \right) -
\left(  \log q_{j-1} + a_j \htheta_{j-1} - \psi(\htheta_{j-1}) \right)
\end{equation}
for any $a \in \R^n$.  So $\frac{\partial G_j}{\partial a_k} = 0$ whenever $|j-k| > 1$ by Lemma \ref{L:some_derivs}.

For the rest of the lemma, just differentiate (\ref{E:G_simpler_again}), use the fact that $\mu = \psi^\prime$ and apply Lemma \ref{L:some_derivs}.  For example, if $j \not= n$ then
\begin{eqnarray*}
\frac{\partial G_j}{\partial a_{j+1}}
&=&  \left(  \frac{1}{q_j} \frac{\partial q_j}{\partial a_{j+1}}
     + \left(a_j - \mu(\htheta_j)\right) \frac{\partial \htheta_j}{\partial a_{j+1}}  \right) -
     \left(  \frac{1}{q_{j-1}} \frac{\partial q_{j-1}}{\partial a_{j+1}}
     + \left(a_j - \mu(\htheta_{j-1})\right) \frac{\partial \htheta_{j-1}}{\partial a_{j+1}}  \right) \\
&=&  \frac{r(a_{j+1})}{q_j} + \left(a_j - \mu(\htheta_j)\right) \frac{r(a_{j+1})}{q_j \mu^\prime(\htheta_j)}
     \left(a_{j+1} - \mu(\htheta_j) \right) \mbox{ by Lemma \ref{L:some_derivs}}  \\
&=&  \frac{r(a_{j+1})}{q_j} \left( 1 + \frac{\left(a_j - \mu(\htheta_j)\right)\left(a_{j+1} - \mu(\htheta_j)\right)} {\mu^\prime(\htheta_j)} \right).
\end{eqnarray*}
\end{proof}

\section*{Acknowledgment}

The author would like to thank Enes Makalic and Daniel F. Schmidt for their many helpful comments on this manuscript and its earlier versions.

\end{document}